\documentclass[final]{siamltex}
\usepackage{amsmath,amsfonts,amssymb}

\newcommand{\F}{\mathbb{F}}

\newcommand{\Cout}{\mathcal{C}_{\textnormal{out}}}
\newcommand{\Cin}{\mathcal{C}_{\textnormal{in}}}

\newtheorem{example}[theorem]{Example}

\title{On Burst Error Correction and Storage Security of Noisy Data\thanks{Research partially supported by Armasuisse and Swiss
National Science Foundation Projects no. 132256 and no. 138080.}} 

\author{Felix Fontein\and
Kyle Marshall\and
Joachim Rosenthal\and
Davide Schipani\and
Anna-Lena Trautmann\thanks{The authors are with the Institute of Mathematics, University of Zurich, Switzerland ({\tt http://www.math.uzh.ch/aa}).}
}

\begin{document}

\maketitle

\begin{abstract}
Secure storage of noisy data for authentication purposes usually involves the use of error correcting codes. We propose a new model scenario involving burst errors and present for that several constructions.
\end{abstract}



\pagestyle{myheadings}
\thispagestyle{plain}

\section{Introduction}

Sensible data, as e.g.\ passwords, are usually stored via the use of a one-way hash function, so that
retrieving the data implies comparing the stored hash with the hash of the new input. In many contexts, though, this simple procedure cannot be used due to the noisy or fuzzy nature of the data.
An outstanding example is the storage of biometric data, e.g.\ in the form of fingerprint, iris, voice, DNA etc., for the purpose of authentication: the data derived from different acquisitions of the same biometric feature can slightly change from each other, and the biometric feature can slightly change itself for different reasons \cite{ul04}. Therefore a certain threshold of tolerance is needed to distinguish legitimate from non legitimate users, but this prevents the standard use of collision resistant hash functions \cite{sc10}. 

This problem has led to the proposal of systems for the secure storage of biometric passwords
(see \cite{tu07} for a selected survey of the literature), which essentially act as ``tolerant'' hash functions.
The idea behind most of these methods is a combined use of error correcting codes and hash functions, whose model is the fuzzy commitment scheme \cite{ju99}. 

A scheme which is apparently just the dual of this is the syndrome
fuzzy hashing construction: in \cite{FUZZY} we showed that
it offers several advantages with respect to the fuzzy
commitment scheme, in particular as far as information leakage is concerned.

The fuzzy commitment scheme has later been generalized to other types of metrics, such as the set difference metric \cite{ju06} and the edit distance metric \cite{do08}.

In particular, the fuzzy vault \cite{ju06} uses polynomial interpolation in order to allow authentication based on the matching of a sufficient number of features, while the fuzzy extractor \cite{do08} is a further generalization which combines the previous constructions with particular objects called random extractors. These make the previous schemes stronger with respect to information leakage, although they cannot prevent it \cite{bu08,do05}.

The choice of one scheme rather than another depends not only on the application or model scenario, but also on important issues, like for example privacy concerns, as already mentioned, or suitability for implementation, as we already discussed in \cite{FUZZY, sc10}.

In this paper we will slightly modify the model scenario: we will tolerate up to, say, $\ell$ error bursts of maximum length $b$, and possibly assume also the presence of other random errors.
This assumption is realistic in many contexts, where errors are likely to appear in bursts, and leads us to focus on and take advantage of burst error correcting codes. We will investigate different choices of codes, depending on the type and the dimension of the bursts, and elaborate on their key features, like error correcting capability and decoding complexity. 

The structure of the paper is the following: in Section $2$ we review the syndrome fuzzy hashing construction, which can be taken as our model scheme. 
The rest of the paper is devoted to solutions using burst error-correcting codes: we deal with base field representations of Reed-Solomon codes in Section $3$ and with concatenated codes in Section $4$. Lastly Section $5$ presents some conclusions.

\section{Syndrome fuzzy hashing and burst error correction}
We review here the syndrome fuzzy hashing construction, as presented in \cite{FUZZY}.

Suppose we need a tolerance of $e$ errors, then an $[n,k]$-linear block code $C \subset\mathbb{F}_q^n$, able to correct $e$ errors, is selected,
and it is described through its $r \times n$ parity-check matrix $H$, with $r = n-k$.
Given a data vector $x$ to be stored, the pair $(H_a(x),Hx)$ is used to represent $x$, were $H_a$ is a given hash function.
When another vector $y$ is acquired and is compared with $x$, the value $Hx-Hy=H(x-y)=Hv$ is computed, that coincides with the syndrome associated to the difference vector $v = x-y$.
Then, syndrome decoding is applied on $Hv$, according to the chosen code $C$.
If $d(x,y) \le e$, then $v$, which has Hamming weight equal to $d(x,y)$, corresponds
to a correctable error vector. So, syndrome decoding succeeds and correctly results in $v$.
Then, starting from $v$ and $y$, $x$ can be computed, as well as $H_a(x)$.
The latter coincides with the stored value, so authentication succeeds.
Otherwise, syndrome decoding fails or reports $w \ne v$. In such case, $x' = w+y \ne x$
and $H_a(x') \ne H_a(x)$ is obtained, and authentication fails.

As stated before, this construction is the dual version of the fuzzy commitment scheme, but offers better security in terms of reduced information leakage.

One can transform an $[n,k]$-linear block code into a two-dimensional (array) code by writing the codewords into an $n_{1}\times n_{2}$-array such that $n_{1}n_{2}=n$. A simple way of doing so is writing the vector entries into the array row by row. For syndrome decoding such a two-dimensional code is transformed back into the vector representation of the original linear block code and can be used as explained before.

\begin{definition}
\begin{enumerate}
\item
A \emph{one-dimensional burst (error) of length $m$} is a vector of length $m$ over $\F_{q}$, such that the first and last entries are non-zero.
\item
A \emph{two-dimensional rectangular burst (error) of size $m\times m'$} is an $m\times m'$- matrix over $\F_{q}$, such that the first and last columns and the first and last rows each contain a non-zero element.
\end{enumerate}
\end{definition}

Note that a one-dimensional burst may occur horizontally or vertically in a two-dimensional code. 

Burst errors of this type appear in many real-life applications, so that in many contexts the syndrome fuzzy hashing construction can exploit the power of burst error correcting codes. 

\section{Base field representations of Reed-Solomon codes}

Reed-Solomon codes are optimal codes defined over an extension field $\F_{q^m}$. When the entries are expanded over the base field, one gets a code over $\F_q$ that may be used for burst error correction in different ways.

\begin{definition}
Consider the extension field $\F_{q^m}$ and $n=q^m-1$ distinct elements $x_1,\dots,x_n \in \F_{q^m}$, and let $0 < k < n$. Then
\[RS=\{(f(x_1),\dots,f(x_n)) \mid f\in \F_{q^m}[x], \deg(f)< k\} \subseteq \F_{q^m}^n\]
is called a \emph{Reed-Solomon} (RS) code. It has length $n$, dimension $k$ and minimum distance $d=n-k+1$, which means that it is optimal (MDS). 
\end{definition}

\subsection*{Expanding as  row vectors (Construction~I)}

When expanding the elements of $\F_{q^m}$ in a Reed-Solomon code as elements of $\F_q^m$ via an $\F_q$-isomorphism $\F_{q^k}\cong \F_{q}^k$, one gets a $q$-ary linear code of length $n'=nm=m(q^m-1)$ and dimension $k'=mk$. The minimum distance is $d'\geq d$ but the burst-error correction capability is higher than the random error correction capability.
In fact, since at most $\lfloor\frac{n-k}{2}\rfloor$ symbols from $\F_{q^m}$ can be corrupted, it is readily seen that such a code is able to correct any
\begin{itemize}
 \item  single burst of length $\leq m(\lfloor\frac{n-k}{2}\rfloor-1)+1$, or
 \item  any $\ell$ many bursts of length $\leq m(\lfloor\frac{n-k}{2\ell}\rfloor-1)+1$.
\end{itemize}

In the binary case, this expansion of codes over $\F_{2^m}$ into binary codes can be found \ in \cite[Ch. 10 \S 5]{ma77}. There one can also find that adding a parity check bit in the end of each $m$-vector results in a code of length $n'=(m+1)(2^m-1)$, dimension $k'=mk$ and minimum distance $d'\geq 2d$. The analogue holds for $q$-ary codes. The higher minimum distance means that one can correct more random errors with this variation of the construction.

\subsection*{Expanding as matrices (Constructions~II and III)}
RS codes in $q$-ary representation can also be exploited if one wants to compare two data matrices (or higher dimensional data), if the difference pattern is likely to consist of two-dimensional bursts. 

To obtain a two-dimensional $q$-ary representation of the Reed-Solomon codes, it may be convenient to write each symbol from $\F_{q^m}$ into a rectangle or a square array. This can be done in different ways, e.g.\ row-wise, column-wise, in spiral shapes etc., which does not make a difference in the error-correction capability. Moreover, one can write the RS code itself (row-wise) into a $n_1\times n_2$-matrix such that $n_1n_2=n$.
If $m$ is a square and we expand the symbols from $\F_{q^m}$ into squares of length $\sqrt{m}$, then one gets a $q$-ary array code of length $n_1\sqrt{m}\times n_2\sqrt{m}$. 
\begin{theorem}\label{thm1}
A $q$-ary RS code with the parameters from above is able to correct 
\begin{itemize}
 \item any single square burst of area $ \leq\Bigl(\sqrt{m}\Bigl\lfloor \sqrt{\frac{n-k}{2}}-1\Bigr\rfloor + 1\Bigr)^2$, or
 \item any $\ell$ many square bursts of area $\leq\Bigl( \sqrt{m} \Bigl\lfloor\sqrt{\frac{n-k}{2\ell}}-1\Bigr\rfloor+1\Bigr)^2$, or
 \item any single one-dimensional burst of length $\leq \sqrt{m}(\lfloor\frac{n-k}{2}\rfloor-1)+1$, or
 \item any $\ell$ many one-dimensional bursts of length $\leq \sqrt{m}(\lfloor\frac{n-k}{2\ell}\rfloor-1)+1$.
\end{itemize}
\end{theorem}
\begin{proof}
The one-dimensional burst error correction capability is the same as in the row vector expansion case, thus it remains to show the square burst case. 


Let $x$ denote the side length of the burst. The burst corrupts the most extension field symbols of the code when e.g.\ the entry of its upper left corner lies on the entry of the lower right corner of some expanded symbol. Then the burst corrupts at most
\[\left(\left\lceil\frac{x-1}{\sqrt{m}}\right\rceil +1\right)^2\]
extension field symbols. Since the number of corrupted extension field elements has to be less than or equal to the error-correction capability of the RS code, it has to hold that
\[\left(\left\lceil\frac{x-1}{\sqrt{m}}\right\rceil +1 \right)^2 \leq \frac{n-k}{2} \iff  x\leq \sqrt{m} \biggl\lfloor\sqrt{\frac{n-k}{2}} -1 \biggr\rfloor + 1\]
which implies the formula. In the case that there are $\ell$ square bursts it has to hold that
\[\ell\left(\left\lceil\frac{x-1}{\sqrt{m}}\right\rceil +1 \right)^2 \leq  \frac{n-k}{2}  \iff  x\leq \sqrt{m} \biggl\lfloor\sqrt{\frac{n-k}{2\ell}}-1 \biggr\rfloor +1 .\]%
\end{proof}

Note, that in analogy to Construction~I, one can again add a parity check element to all vector expansions of the extension field elements to achieve a larger minimum distance.

Another way of expanding the elements of $\F_{q^m}=\F_{q}[\alpha]$ is using companion matrices. For this let $p(x)\in \F_q[x]$ be the minimal polynomial of $\alpha$ and $P\in \F_q^{m\times m}$ its companion matrix. Then $\F_{q^m} \cong \F_q[P]$ and we can represent the extension field elements by $q$-ary matrices. 
If the RS code is written in an $n_1\times n_2$-matrix like above, then this results in a $q$-ary array code of length $n_1 m\times n_2 m$. 
\begin{theorem}
A $q$-ary RS code expanded with companion matrices with the parameters from above is able to correct 
\begin{itemize}
 \item any single square burst of area $ \leq\Bigl(m\Bigl\lfloor \sqrt{\frac{n-k}{2}}-1\Bigr\rfloor + 1\Bigr)^2$, or
 \item any $\ell$ many square bursts of area $\leq\Bigl( m \Bigl\lfloor\sqrt{\frac{n-k}{2\ell}}-1\Bigr\rfloor+1\Bigr)^2$, or
 \item any single one-dimensional burst of length $\leq m(\lfloor\frac{n-k}{2}\rfloor-1)+1$, or
 \item any $\ell$ many bursts of length $\leq m(\lfloor\frac{n-k}{2\ell}\rfloor-1)+1$.
\end{itemize}
\end{theorem}
The proof is analogous to the one of Theorem \ref{thm1} and hence omitted. 

It can be noticed that this second expansion can offer a comparable burst error correction capability with respect to the other expansion (considering $q$-ary expansions of comparable length), but allowing to start with a RS code of smaller length and having lower decoding complexity. The rate though would be relatively smaller, which may reduce the security of the scheme.


\section{Concatenated codes}

When the error pattern is a mix of burst and random errors, concatenated codes provide the desired flexibility.

\begin{definition}\label{concatdef}
	Let $\Cout\colon (\F_{q^k})^K\rightarrow (\F_{q^k})^N$ and $\Cin^{(i)}\colon \F_{q}^k\rightarrow \F_q^n$ for $i=1,\dots,N$ be encoding functions. The \emph{concatenated code} $\mathcal{C} = \Cin^{(1)}\|\Cin^{(2)}\|\ldots \|\Cin^{(N)}$ is the $q$-ary code of length $Nn$ and dimension $Kk$ given by 
	$$ \mathcal{C} = \bigl\{(\Cin^{(1)}(y_1), \Cin^{(2)}(y_2),\ldots ,\Cin^{(N)}(y_N))\bigm| (y_1,y_2,\ldots,y_N) = \Cout(x), x\in (\F_{q^k})^K\bigr\}  .$$
\end{definition}

Here we assume the isomorphism $\F_{q^k}\cong \F_{q}^k$, so that the maps $\Cin$ agree in the domain. The codes $\Cin^{(1)},\ldots \Cin^{(N)}$ are called the \textit{inner codes} and $\Cout$ is called the \textit{outer code}. 
The decoding process is a two-step process in which the inner codes are decoded first, followed by the outer code. Suppose that the inner codes are identical, that is $$\mathcal{C} = \underbrace{C \| C \| \ldots \| C}_{n\textnormal{ times}},$$ and we can represent an element $c\in\mathcal{C}$ as \begin{equation}\label{concatrep} c = (c_{1,1}, \ldots, c_{1,n}\|\ldots \|c_{N,1}, \ldots, c_{N,n}).\end{equation} Assume that $C$ and $\Cout$ have efficient decoding procedures. 
Suppose that the decoding procedure for $C$ can correct $t$ errors and the decoding procedure for $\Cout$ can correct $s$ errors. If fewer than $t+1$ errors occur in the coordinates of the $i$th copy of $C$ then the inner decoding step can correct these errors and the block would no longer count as an error in the outer code. If the number of errors is greater than $t$ then the $n$ symbols corresponding to the $i$th copy of $C$ count as only one error to the outer decoder. One consequence of this multi-layer decoding process is that $\mathcal{C}$ can correct any burst of length smaller than $n(s-1)+2t+1$, beside the possibility of correcting other random errors. Alternatively, it is not difficult to see that $\mathcal{C}$ also has the capability to correct multiple burst errors, provided their lengths are not too long.

\subsection*{Construction~IV}
In the two-dimensional case, concatenated codes can be interleaved in ways that have advantages against different types of error patterns. Let $b \mid N$ and $a \mid n$. Arrange the elements of $c$ in \eqref{concatrep} in the $\frac{Na}{b}\times \frac{nb}{a}$ array given by 
\begin{equation*}\label{concatinterleave1}\left(
\begin{array}{ccc|c|ccc}
	c_{1,1}  & \ldots & c_{b,1} &  \ldots & c_{1,n/a} & \ldots & c_{b,n/a} \\
	c_{b+1,1}  & \ldots & c_{2b,1} &  \ldots & c_{b+1,n/a} & \ldots & c_{2b,n/a} \\
	& \vdots & & \vdots &  & \vdots & \\
	c_{N-b+1,1} &  \ldots & c_{N,1} & \ldots & c_{N-b+1,n/a} & \ldots & c_{N,n/a} \\ \hline & \vdots & &\ddots & &\vdots & \\ \hline
	c_{1,\frac{(a-1)n}{a}+1}  & \ldots & c_{b,\frac{(a-1)n}{a}+1} &  \ldots & c_{1,n}  & \ldots & c_{b,n} \\
	c_{b+1,\frac{(a-1)n}{a}+1}  & \ldots & c_{2b,\frac{(a-1)n}{a}+1} &  \ldots & c_{b+1,n}  & \ldots & c_{2b,n} \\
	& \vdots & & \vdots & \vdots & \\
	c_{N-b+1,\frac{(a-1)n}{a}+1}  & \ldots & c_{N,\frac{(a-1)n}{a}+1} & \ldots & c_{N-b+1,n} & \ldots & c_{N,n} \\
\end{array}
\right)  .\end{equation*}
Using this interleaving pattern, any submatrix of size $\frac{N}{b}\times b$ contains only one symbol from each inner code. It is straightforward to see that this interleaving scheme can correct at least $t$ rectangular bursts of size $\frac{N}{b}\times b$, since it can correct $t$ such bursts using the inner codes alone. Additionally, this interleaving pattern can correct at least $s$ random errors.

\subsection*{Construction~V}
Without interleaving, one could consider the following pattern. Let $a \mid N$ and $b \mid n$. Arrange the elements of $c$ in \eqref{concatrep} in the $\frac{Nn}{ab}\times a b$ array given by 
\begin{equation*}\label{concatnointerleave}\left(
\begin{array}{ccc|c|ccc}
	c_{1,1}  & \ldots & c_{1,b} &  \ldots & c_{a,1} & \ldots & c_{a,b} \\
	c_{1,b+1}  & \ldots & c_{1,2b} &  \ldots & c_{a,b+1} & \ldots & c_{a,2b} \\
	& \vdots & & \vdots &  & \vdots & \\
	c_{1,n-b+1} &  \ldots & c_{1,n} & \ldots & c_{a,n-b+1} & \ldots & c_{a,n} \\ \hline
        & \vdots & &\ddots & &\vdots & \\ \hline
	c_{N-a+1,1}  & \ldots & c_{N-a+1,b} &  \ldots & c_{N,1}  & \ldots & c_{N,b} \\
	c_{N-a+1,b+1}  & \ldots & c_{N-a+1,2b} &  \ldots & c_{N,b+1}  & \ldots & c_{N,2b} \\
	& \vdots & & \vdots & \vdots & \\
	c_{N-a+1,n-b+1}  & \ldots & c_{N-a+1,n} & \ldots & c_{N,n-b+1} & \ldots & c_{N,n} \\
\end{array}
\right) .\end{equation*} Using this pattern, any single burst of maximal size $((s_1 - 1) \frac{n}{b} + 1) \times
((s_2 - 1) b + 1)$ can be corrected by the outer code, where $s_1, s_2 \in \mathbb{N}$ with $s_1 s_2
\le s$, and as long as each submatrix of size~$\frac{n}{b} \times b$ which does not collide with a burst
contains at most~$t$ errors, all errors can be corrected. Therefore, this scheme works well if there
is at most one larger burst, or few smaller ones, and a lot of random errors spread more or less uniformly in the matrix. 

Note that Constructions~II and III can be seen as special cases of Construction~V, where the inner
code is the trivial code with $b = \sqrt{m}$ for Construction~II, and with the encoding $\F_{q^m} \cong
\F_q[P] \subseteq \F_q^{m \times m}$ and $b = m$ for Construction~III.

\subsection*{Construction~VI}
Consider another interleaving pattern given by the $n\times N$ array (say $N \geq n$): 
\begin{equation*}\label{concatinterleave2}\left(
\begin{array}{cccccc}
	c_{1,1} & c_{2,1} & c_{3,1} & c_{4,1} & \ldots & c_{N,1} \\
	c_{N,2} & c_{1,2} & c_{2,2} & c_{3,2} & \ldots & c_{N-1,2} \\
	& & \vdots & & & \\
	c_{N-n+2,n} & c_{N-n+3,n} & c_{N-n+4,n} & c_{N-n+5,n} & \ldots & c_{N-n+1,n}
\end{array}
\right)\end{equation*}
Each inner code is interleaved diagonally, and it is clear from the construction that any burst of size $1\times n$ or $n\times 1$ corrupts only one symbol from each inner code. Additionally, an error pattern consisting of a diagonal burst would corrupt an entire inner codeword but be treated as a single error to the outer decoder.

Using a fuzzy scheme in which the witness is prone to burst errors as well as random errors, an
interleaving scheme like that in Construction~IV may be advantageous. In case less burst errors
happen, but more random errors, a scheme like the ones in Construction~V or VI might be better
suited.

\section{Conclusions}

One of the big advantages of using burst error correcting codes relies in the decoding complexity, since we are correcting errors in words of a certain length but using decoding procedures for codes of smaller lengths. For example
in Constructions~I, II and III the decoding complexity is dominated by the Reed-Solomon decoder in the extension field. Using the standard Gorenstein-Peterson-Zierler decoding procedure
 this is given by $$O\Bigl(n \cdot \frac{n-k}{2}\Bigr)= O\Bigl((q^m-1)\frac{q^m-1-k}{2}\Bigr) = O(q^{2m})$$
operations in $\F_{q^m}$. More recent algorithms for decoding cyclic codes \cite{ISIT} can even achieve a complexity of 
$$O\Bigl(\sqrt{n}\log n\cdot\frac{n-k}{2}\Bigr) = O(mq^{\frac{3}{2}m}).$$
Note, that in Construction~III  the conversion of the elements of $\F_q[P]$ to elements of $\F_q[\alpha]$ is more complex then the conversion of elements of $\F_{q}^{m}$ to elements of $\F_{q^{m}}$ needed in Constructions~I and II, but the complexity is still polynomial in $m$ and thus does not change the overall complexity in the Big-O notation. 

The decoding complexity for Constructions~IV, V and VI depends on the choices of the inner and outer codes. If for example you choose a RS-code over $\F_{q^{k}}$ for the outer code and a $q$-ary narrow-sense primitive BCH-code of length $n=q^{m}-1$ for the inner code, then the decoding complexity is $O(N n t)=O((q^{k}-1) (q^{m}-1) t)=O( q^{k+m}t)$
operations in $\F_{q^{m}}$, 
where $t$ is the error-correction capability of the inner code, and $O(q^{2k})$ operations in $\F_{q^{k}}$. 


Depending on the application, any of the six presented constructions can be advantageous for burst error correction in the syndrome fuzzy hashing scenario. The differences appear in the type and number of correctable errors as well as in the decoding complexity for given code word size in the $q$-ary expansion. Naturally, this list of constructions is not complete and there exist other codes that can be useful for burst-error correction in the storage of noisy data.

We want to conclude with a list of recommendations for which situations the above constructions can
be used:
\begin{itemize}
  \item Constructions~I and II are relatively fast to decode (as the inner code is trivial), while
    having a good burst error correction capability. They can also correct a few random errors, but
    not too many.
  \item 
Construction~III can correct bursts similarly well as Construction~II, and compensates a relatively lower rate with a faster decoding procedure, which can be relevant when
    deploying the syndrome fuzzy hashing scheme on embedded devices.
  \item Construction~IV works well with several and larger bursts, but not too many random errors.
  \item Construction~V works well with a few large bursts and many uniformly enough distributed
    random errors.
  \item Construction~VI is well suited for rectangular bursts of size $k \times \ell$ with $k$ much
    smaller than $\ell$ (or vice versa), and for random errors.
\end{itemize}


\begin{thebibliography}{10}

\bibitem{FUZZY}
M.~Baldi, M.~Bianchi, F.~Chiaraluce, J.~Rosenthal, D.~Schipani.
\newblock On fuzzy syndrome hashing with LDPC coding.
\newblock In {\em Proc. 4th ACM International Symposium on Applied Sciences in Biomedical and Communication Technologies (ISABEL)}, pp. 1-5, Barcelona, 2011.

\bibitem{bu08}
I.~Buhan, J.~Doumen, P. Hartel.
\newblock Controlling leakage of biometric information using dithering.
\newblock In {\em Proc. EUSIPCO 2008}, Lausanne, 2008.


\bibitem{do08}
Y.~Dodis, R.~Ostrovsky, L.~Reyzin, A.~Smith.
\newblock Fuzzy extractors: how to generate strong keys from biometrics and other noisy data.
\newblock {\em SIAM J. Comp.}, vol.~38, no.~1, pp. 97--139, 2008.


\bibitem{do05}
Y.~Dodis, A.~Smith.
\newblock Correcting errors without leaking partial information.
\newblock In {\em Proc. ACM STOC 2005}, pp. 654--663, Baltimore MD, 2005.



\bibitem{ju06}
A.~Juels, M.~Sudan.
\newblock A fuzzy vault scheme.
\newblock {\em Designs, Codes and Cryptography}, vol.~38, no.~2, pp. 237--257, 2006.


\bibitem{ju99}
A.~Juels, M.~Wattenberg.
\newblock A fuzzy commitment scheme.
\newblock In {\em Proc. 6th ACM Conference on Computer and Communications Security},
pp. 28--36, Singapore, 1999.


\bibitem{ma77}
F.J.~MacWilliams, N.J.A.~Sloane.
\newblock The Theory of Error-Correcting Codes.
\newblock North Holland, 1977. 





\bibitem{ISIT} 
D.~Schipani, M.~Elia, J.~Rosenthal.  
\newblock On the Decoding Complexity of Cyclic Codes up to the BCH Bound.
\newblock In {\em Proc. IEEE Int. Symp. on Information Theory}, pp.835--839, St. Petersburg, 2011.


\bibitem{sc10}
D.~Schipani, J.~Rosenthal.
\newblock Coding solutions for the secure biometric storage problem.
\newblock In {\em Proc. IEEE Information Theory Workshop 2010}, Dublin, 2010.

\bibitem{tu07}
P.~Tuyls, B.~Skoric, T.~Kevenaar (Eds.).
\newblock {\em Security with noisy  data}.
\newblock Springer, 2007.

\bibitem{ul04}
U.~Uludag, S.~Pankanti, S.~Prabhakar, A.K.~Jain.
\newblock Biometric cryptosystems: Issues and challenges.
\newblock In {\em Proc. of the IEEE}, vol.~92, no.~6, pp. 948--960, 2004.


\end{thebibliography}
\end{document}